\documentclass[ASNA,twocolumn]{USG}
\usepackage{anyfontsize}
\usepackage{colortbl}
\usepackage{xcolor}
\usepackage{booktabs}
\usepackage{steinmetz}

\newcommand{\sgn}{\operatorname{sgn}}
\newcommand{\Sgn}{\operatorname{Sgn}}
\newcommand{\spow}[2]{\left\lceil #1 \right\rfloor^{#2}}

\graphicspath{{./images/}}

\articletype{SHORT ARTICLE}

\received{\today}
\revised{-}
\accepted{-}
\journal{Journal}
\volume{0}
\copyyear{2026}
\startpage{1}
\articledoi{}

\begin{document}

\title{Zero-Crossing Time Lower Bounds for the Homogeneous Super-Twisting Algorithm}
\transtitle{Zero-Crossing Time Lower Bounds for the Homogeneous Super-Twisting Algorithm}
\subtranstitle{trans-subtitle}
\author[1]{Diego Guti\'errez-Oribio}[https://orcid.org/0000-0002-5895-3416]
\author[1]{Ioannis Stefanou}[https://orcid.org/0000-0002-4552-7717]

\authormark{Guti\'errez-Oribio and Stefanou}
\titlemark{Zero-Crossing Time Lower Bounds for the Homogeneous Super-Twisting Algorithm}

\address[1]{\orgdiv{IMSIA (UMR 9219), CNRS, EDF, CEA, }\orgname{ENSTA Paris, Institut Polytechnique de Paris, }%
\orgaddress{\state{Palaiseau, }\country{France}}}

\corres{Diego Guti\'errez-Oribio  (\email{diego.gutierrez@ensta.fr}) ~|~ Ioannis Stefanou (\email{ioannis.stefanou@ensta.fr})}

\editor{\textbf{Academic Editor:} - ~|~ \textbf{Guest Editor:} -}

\presentaddress{IMSIA (UMR 9219), CNRS, EDF, CEA, ENSTA Paris, Institut Polytechnique de Paris, Palaiseau, France.}

\fundingInfo{European Research Council's (ERC) grant INJECT, no. 101087771, doi: 10.3030/101087771.}

\keywords{Sliding mode control | Nonlinear Control | Robust Control | Super-twisting algorithm | Zero-crossing estimates}

\transkeywords{Sliding mode control | Nonlinear Control | Robust Control | Super-twisting algorithm | Zero-crossing estimates}

\abstract[ABSTRACT]{
High-frequency switching and chattering are well-known problems in sliding-mode control. Although the super-twisting algorithm produces a continuous control input, the sliding variable may still cross zero repeatedly near the sliding manifold. This paper studies how fast these crossings can occur and whether they can accumulate in finite time for the scalar homogeneous super-twisting algorithm. Two lower bounds on the time required to return to zero are derived, together with a bound on the time needed to complete a full excursion. These results give minimum separation times and upper bounds on the number of crossings associated with excursions above a prescribed amplitude. Although zero crossings may accumulate in finite time, the amplitudes of the complete excursions considered here must vanish and satisfy a summability condition. Thus, crossings associated with excursions above any prescribed nonzero amplitude cannot accumulate in finite time. When the perturbation-rate bound is zero, the estimates recover the scaling imposed by homogeneity. For the classical super-twisting algorithm, the square-root scaling also holds for nonzero perturbation-rate bounds. Numerical simulations illustrate the results.
}

\transabstract[transABSTRACT]{
High-frequency switching and chattering are well-known problems in sliding-mode control. Although the super-twisting algorithm produces a continuous control input, the sliding variable may still cross zero repeatedly near the sliding manifold. This paper studies how fast these crossings can occur and whether they can accumulate in finite time for the scalar homogeneous super-twisting algorithm. Two lower bounds on the time required to return to zero are derived, together with a bound on the time needed to complete a full excursion. These results give minimum separation times and upper bounds on the number of crossings associated with excursions above a prescribed amplitude. Although zero crossings may accumulate in finite time, the amplitudes of the complete excursions considered here must vanish and satisfy a summability condition. Thus, crossings associated with excursions above any prescribed nonzero amplitude cannot accumulate in finite time. When the perturbation-rate bound is zero, the estimates recover the scaling imposed by homogeneity. For the classical super-twisting algorithm, the square-root scaling also holds for nonzero perturbation-rate bounds. Numerical simulations illustrate the results.
}

\abbr{}
\contributed{}
\dedicated{}
\copyright{}

\maketitle

\section{Introduction}

The super-twisting algorithm (STA) is a second-order sliding-mode control for relative-degree-one systems \cite{b:levant_ijc1993}. It preserves continuity of the applied control while establishing the sliding motion in finite time under the classical robustness conditions. Homogeneity provides a unified construction of the STA and its continuous extensions \cite{b:levant_automatica2005,b:LEVANT2007576,b:Mathey-Moreno-2024}, whereas strict Lyapunov functions provide constructive stability conditions and convergence-time estimates \cite{b:PolyakovPoznyak2009ReachingTime,b:6144710}. Reaching-time analyses have also been developed directly from the STA trajectories. In particular, an exact reaching-time function and a semi-explicit trajectory representation are available for the unperturbed classical STA, where the zero structure of the transformed trajectory is treated explicitly \cite{b:SeeberHornFridman2018ReachingTime}.

High-frequency switching and chattering are among the main practical and theoretical concerns associated with sliding-mode control. Although the STA produces a continuous applied control, this does not eliminate fast activity of the sliding variable: repeated sign-changing zero crossings may still occur near the sliding manifold. Moreover, actuator dynamics, sampling, delays, quantization, and other parasitic effects may generate sustained oscillations in practical implementations \cite{b:8896941,b:SEEBER2025112027}. Describing-function and harmonic-balance analyses quantify the amplitude and frequency of such implementation-induced limit cycles once the neglected dynamics are explicitly included \cite{b:BoikoFridman2005,b:RosalesFreidovichCastillo2023,b:PEREZVENTURA2019175}. These results provide valuable implementation-oriented chattering estimates, but they do not characterize the temporal structure of the zero crossings generated by the continuous-time STA trajectory itself.

Repeated oscillations across the sliding manifold generate a sequence of sign-changing zero crossings. The temporal spacing of these crossings therefore characterizes the time scale of the oscillatory motion. The zero-crossing instants can be regarded as events generated by the closed-loop dynamics. Positive lower bounds on inter-event times are central in hybrid and event-triggered control because they exclude Zeno executions \cite{b:GoebelSanfeliceTeel2012,b:YuBarbotZeno2011}. Event-triggered sliding-mode and super-twisting designs obtain such bounds for controller updates generated by an imposed triggering mechanism \cite{b:LiEtAl2020EventTriggeredSTA,b:KumarMohantyKumar2020EventTriggeredSTA,b:DouDingYu2023EventTriggeredSOSM}. In the present setting, however, the events are not externally imposed: they are the sign-changing zero crossings generated by the continuous-time STA trajectory itself. Their separation must therefore be inferred directly from the intrinsic closed-loop dynamics rather than from a designed sampling or triggering rule. Zero-crossing instants have also been used to drive STA gain adaptation \cite{b:AndersonEtAl2023ZeroCrossing}. In that setting, the crossings determine the gain update. 

The objective of this work is to quantify fast zero-crossing activity in the STA through excursion-dependent lower bounds on the return time. This differs from the exact unperturbed analysis in \cite{b:SeeberHornFridman2018ReachingTime}, where the zero structure is resolved as part of the total reaching-time computation. Two complementary bounds are derived: the first is uniform with respect to the uncertain positive input coefficient, while the second exploits its known upper bound and can provide a sharper estimate. Neither requires measuring or reconstructing the matched perturbation, the disturbance-compensated coordinate, or the crossing velocity.

The qualitative impossibility of accumulating excursions bounded away from zero on a finite interval already follows from continuity, but the contribution here is quantitative. The bounds provide an explicit minimum separation for crossings associated with a prescribed excursion threshold, together with finite-count and cumulative estimates. For a perturbation with zero rate bound, they recover the scaling imposed by homogeneity, while for the classical super-twisting algorithm the square-root dependence on the excursion amplitude remains valid for arbitrary admissible perturbation rates. The cumulative result further yields a summability property for the excursion amplitudes and characterizes the effective input at a possible accumulation point.

The remainder of the paper is organized as follows. Section~\ref{sec:problem} introduces the plant--controller interconnection and the zero-crossing intervals. Section~\ref{sec:main} establishes the excursion-dependent bounds and their consequences. Section~\ref{sec:numerics} presents the numerical illustration. Section~\ref{sec:discussion} positions the result with respect to stability, chattering, and event-triggered analyses. Concluding remarks are given in Section~\ref{sec:conclusion}.

\section{Problem setting}
\label{sec:problem}

Consider the scalar plant
\begin{equation}
\label{eq:plant}
\dot \sigma(t)=a(t)\bigl(u(t)+\phi(t)\bigr),
\end{equation}
where $\sigma$ is the state, $u$ is the control input and $\phi$ is a matched perturbation. The homogeneous STA is defined by
\begin{equation}
\label{eq:controller}
\begin{aligned}
u(t)=-k_1\spow{\sigma(t)}{\alpha}+\nu(t),\qquad
\dot\nu(t)=-k_2\spow{\sigma(t)}{\beta},
\end{aligned}
\end{equation}
where $\nu$ is the controller integral state, $k_1,k_2>0$, $\alpha=\frac{1}{1-\ell}$, $\beta=\frac{1+\ell}{1-\ell}$, and $\ell\in[-1,0]$ \cite{b:LEVANT2007576,b:Mathey-Moreno-2024}. For $\gamma>0$, the signed power is defined as $\spow{s}{\gamma}:=|s|^\gamma\sgn(s)$. For $\gamma=0$, $\spow{s}{0}:=\sgn(s)$ for $s\neq0$, with the Filippov extension $\Sgn(0)=[-1,1]$ at the origin.

The input coefficient is assumed measurable and satisfies
\begin{equation}
\label{eq:a-bounds}
0<a_m\le a(t)\le a_M
\end{equation}
for almost all $t\ge0$. Thus, its magnitude may vary with time, but its sign is known and strictly positive. The perturbation $\phi$ is locally absolutely continuous and satisfies
\begin{equation}
\label{eq:phi-derivative-bound}
|\dot\phi(t)|\le\delta
\end{equation}
for almost all $t\ge0$, and $\delta\ge0$ is known.

Introducing the disturbance-compensated coordinate $z(t):=\nu(t)+\phi(t)$ and denoting $d(t):=\dot\phi(t)$, the closed-loop dynamics are
\begin{equation}
\label{eq:closed-loop}
\begin{aligned}
\dot\sigma=a(t)\bigl(-k_1\spow{\sigma}{\alpha}+z\bigr),\qquad
\dot z=-k_2\spow{\sigma}{\beta}+d(t).
\end{aligned}
\end{equation}
For $\ell=-1$, one has $\alpha=1/2$ and $\beta=0$, and the
second equation is understood in the Filippov sense \cite{b:filippov} as $\dot z(t)\in -k_2\operatorname{Sgn}(\sigma(t))+d(t)$. Consequently, every admissible selection satisfies $|\dot z(t)|\leq k_2+\delta$ almost everywhere. Values $\ell\in(-1,0]$ give the continuous homogeneous extensions, with the endpoint $\ell=0$ corresponding to a linear proportional--integral structure.

Stability and robustness conditions for the classical STA and its homogeneous extensions are available in \cite{b:levant_automatica2005,b:LEVANT2007576,b:Mathey-Moreno-2024,b:6144710,b:Seeber-Horn}. The present analysis is conditional on the existence of a forward solution of \eqref{eq:closed-loop}. It does not establish new stability or robustness conditions. Its purpose is to characterize the zero-crossing time associated with a given excursion of that trajectory.

\section{Main results}
\label{sec:main}

Define the effective input acting on the sliding variable as
\begin{equation}
\label{eq:w-definition}
w(t):=u(t)+\phi(t)
=-k_1\spow{\sigma(t)}{\alpha}+z(t).
\end{equation}
Then
\begin{equation}
\label{eq:sigma-w}
\dot\sigma(t)=a(t)w(t)
\end{equation}
for almost all $t$. The variables $z$ and $w$ are continuous along every forward solution. In particular, at any zero $t_i$ of $\sigma$,
\begin{equation}
\label{eq:w-z-at-zero}
w(t_i)=u(t_i)+\phi(t_i)=z(t_i).
\end{equation}

Under the preceding assumptions, $\dot\sigma$ is locally essentially bounded along every forward solution. Indeed, $\sigma$ is continuous on compact intervals. For $\beta>0$, the second equation of \eqref{eq:closed-loop} then implies local essential boundedness of $\dot z$. For $\beta=0$, the Filippov right-hand side is bounded by $k_2+\delta$. Consequently, $z$, $w$, and $a w$ are bounded on compact intervals.

\begin{definition}[Sign-changing zero crossing]
\label{def:crossing}
A time $t_i$ is an upward zero crossing of $\sigma$ if $\sigma(t_i)=0$ and there exists $r_i>0$ such that $\sigma(t)<0$, $t\in(t_i-r_i,t_i)$, $\sigma(t)>0$, $t\in(t_i,t_i+r_i)$. A downward zero crossing is defined analogously, with the signs reversed. Tangential contacts and intervals on which $\sigma$ remains identically zero are not counted as zero crossings (no sign change).
\end{definition}

\begin{definition}[Zero-crossing interval and excursion amplitude]
\label{def:crossing-interval}
Let $t_i<t_{i+1}$ be two consecutive sign-changing zero crossings. The interval $[t_i,t_{i+1}]$ is called a zero-crossing interval if there exists $q_i\in\{-1,1\}$ such that $q_i\sigma(t)>0$, $t\in(t_i,t_{i+1})$. Its excursion amplitude is $A_i:=\max_{t\in[t_i,t_{i+1}]}|\sigma(t)|$. The sign $q_i$ is the sign of $\sigma$ between the two crossings.
\end{definition}

Since $\sigma$ is continuous and changes sign at both endpoints, every zero-crossing interval satisfies $A_i>0$. Moreover, the compactness of $[t_i,t_{i+1}]$ guarantees the existence of at least one time at which $|\sigma|$ reaches $A_i$.

The first result relates the direction of a zero crossing to the sign of the effective input.

\begin{proposition}[Sign of the effective input at a zero crossing]
\label{prop:crossing-sign}
Let $t_i$ be a sign-changing zero crossing. Then $z(t_i)\ge0$ at an upward crossing and $z(t_i)\le0$ at a downward crossing.
\end{proposition}

\begin{proof}
At a zero crossing, \eqref{eq:w-z-at-zero} gives $z(t_i)=w(t_i)$. If an upward crossing had $w(t_i)<0$, continuity of $w$ and $a(t)\ge a_m>0$ would make $\sigma$ strictly decrease near $t_i$, which is a contradiction. Thus $z(t_i)\ge0$. The downward case is analogous.
\end{proof}

\begin{lemma}[Value of the effective input when the excursion amplitude is reached]
\label{lem:maximum-relation}
Consider a zero-crossing interval $[t_i,t_{i+1}]$ with sign $q_i$ and excursion amplitude $A_i$. Let $t_i^A\in(t_i,t_{i+1})$ be a time at which
\begin{equation}
\label{eq:maximum-time}
|\sigma(t_i^A)|=A_i.
\end{equation}
Then, $w(t_i^A)=0$ and
\begin{equation}
\label{eq:z-at-maximum}
q_i z(t_i^A)=k_1A_i^\alpha.
\end{equation}
\end{lemma}

\begin{proof}
Since $q_i\sigma(t_i^A)=A_i$, $q_iw(t_i^A)>0$ would, by continuity of $w$ and $a(t)\ge a_m>0$, make $q_i\sigma$ increase just after $t_i^A$, contradicting the definition of $A_i$. If $q_iw(t_i^A)<0$, the same argument backward in time gives a larger value before $t_i^A$. Hence $w(t_i^A)=0$. Using $\sigma(t_i^A)=q_iA_i$ in \eqref{eq:w-definition} gives \eqref{eq:z-at-maximum}.
\end{proof}

\begin{theorem}[Excursion-dependent zero-crossing time bounds]
\label{thm:main}
Consider a zero-crossing interval $[t_i,t_{i+1}]$ with excursion amplitude $A_i$, and define $M(A):=k_2A^\beta+\delta$. If $t_i^A$ is any time satisfying \eqref{eq:maximum-time}, then the return time satisfies
\begin{equation}
\label{eq:return-time-combined}
t_{i+1}-t_i^A
\ge
\max\{\tau(A_i),\tau_a(A_i)\},
\end{equation}
where $\tau(A):=\frac{k_1A^\alpha}{M(A)}$, and
\begin{equation}
\label{eq:a-dependent-return-bound}
\tau_a(A):=
\sqrt{\frac{2A}{a_M M(A)}}.
\end{equation}
Moreover, the rise time satisfies
\begin{equation}
\label{eq:rise-time-bound}
t_i^A-t_i\ge\tau_r(A_i),
\end{equation}
with
\begin{equation}
\label{eq:tau-r}
\tau_r(A):=
\frac{
\sqrt{k_1^2A^{2\alpha}+2M(A)A/a_M}
-k_1A^\alpha
}{M(A)}.
\end{equation}
Consequently,
\begin{equation}
\label{eq:complete-interval-bound}
t_{i+1}-t_i
\ge
\Theta(A_i)
:=
\tau_r(A_i)+\max\{\tau(A_i),\tau_a(A_i)\}.
\end{equation}
\end{theorem}

\begin{proof}
On $(t_i,t_{i+1})$, $\spow{\sigma}{\beta}=q_i|\sigma|^\beta$, so $q_i\dot z=-k_2|\sigma|^\beta+q_id$ almost everywhere. Proposition~\ref{prop:crossing-sign} gives $q_i z(t_{i+1})\le0$, and \eqref{eq:z-at-maximum} therefore yields
\begin{align}
 k_1A_i^\alpha
\le q_i\bigl(z(t_i^A)-z(t_{i+1})\bigr)
=\int_{t_i^A}^{t_{i+1}}
\bigl(k_2|\sigma(t)|^\beta-q_id(t)\bigr)\,dt.
\label{eq:integral-change}
\end{align}
Using $|\sigma|\le A_i$ and $|d|\le\delta$ gives $t_{i+1}-t_i^A\ge\tau(A_i)$.

For the second return bound, set $s=t-t_i^A$ and $M_i=M(A_i)$. Since $q_i\dot z\ge-M_i$, $q_i z(t)\ge k_1A_i^\alpha-M_i s$, and, using $|\sigma|\le A_i$, $q_iw(t)\ge-M_i s$. If $q_iw(t)<0$, then $a(t)\le a_M$ gives $q_i\dot\sigma=a(t)q_iw\ge a_Mq_iw\ge-a_MM_i s$. If $q_iw(t)\ge0$, the same lower bound is immediate. Thus $q_i\dot\sigma(t)\ge-a_MM_i s$. Integration up to $t_{i+1}$ gives $A_i\le\frac{a_MM_i}{2}(t_{i+1}-t_i^A)^2$, which proves \eqref{eq:a-dependent-return-bound}.

For the rise time, the same bound on $q_i\dot z$ gives $q_i z(t)\le k_1A_i^\alpha+M_i(t_i^A-t)$ on $[t_i,t_i^A]$. If $q_iw(t)\ge0$, then $a(t)\le a_M$ gives $q_i\dot\sigma(t)\le a_M[k_1A_i^\alpha+M_i(t_i^A-t)]$. If $q_iw(t)<0$, this upper bound is immediate. With $T_i^r=t_i^A-t_i$, integration gives $A_i\le a_M\left(k_1A_i^\alpha T_i^r+\frac{M_i}{2}(T_i^r)^2\right)$. Solving the quadratic inequality gives \eqref{eq:rise-time-bound}--\eqref{eq:tau-r}. Adding the rise and return bounds gives \eqref{eq:complete-interval-bound}.
\end{proof}

\begin{remark}[Homogeneous scaling and the role of $a_M$]
\label{rem:homogeneous-scaling}
The homogeneous state weights of system \eqref{eq:closed-loop} are $(r_\sigma,r_z)=(1-\ell,1)$ with degree $\ell$. Since $\beta=2\alpha-1$, when $\delta=0$ the three bounds $\tau$, $\tau_a$, and $\tau_r$ all scale as $A^{-\ell/(1-\ell)}$, which is the amplitude--time scaling imposed by homogeneity. For the classical STA, $\ell=-1$ and $\beta=0$, and the square-root scaling holds for every $\delta\ge0$.

The two return bounds capture different mechanisms. The bound $\tau(A)$ is obtained from the finite change required in $z$ and is uniform over all admissible positive input coefficients. The companion bound $\tau_a(A)$ uses $a_M$ to limit how rapidly $\sigma$ can return from its maximum. In particular, if $\ell>-1$ and $\delta>0$, $\tau(A)\sim\frac{k_1}{\delta}A^\alpha$, $\tau_a(A)\sim\sqrt{\frac{2A}{a_M\delta}}$, $A\to0$, so the $a_M$-dependent bound is asymptotically stronger. Thus, the distinctive property of $\tau$ is its uniformity with respect to $a(t)$, while $\tau_a$ can provide a sharper bound when $a_M$ is used. For comparison, a generic local Lipschitz estimate $|\dot\sigma|\le L$ gives only $t_{i+1}-t_i^A\ge A_i/L$ on the return segment, with a trajectory-dependent constant $L$. The bounds above are explicit in the available system quantities.
\end{remark}

\begin{corollary}[Prescribed excursion threshold]
\label{cor:threshold}
Let $\varepsilon>0$ and consider a zero-crossing interval satisfying $A_i\ge\varepsilon$. Define $t_i^\varepsilon
:=
\inf\bigl\{t\in(t_i,t_{i+1}):|\sigma(t)|=\varepsilon\bigr\}$. Then $t_{i+1}-t_i^\varepsilon
\ge
\overline\tau_\varepsilon
:=
\max\{\tau_\varepsilon,\tau_{a,\varepsilon}\}$, where
$\tau_\varepsilon
:=
\frac{k_1\varepsilon^\alpha}{M(\varepsilon)}$, $\tau_{a,\varepsilon}
:=
\sqrt{\frac{2\varepsilon}
{a_M M(\varepsilon)}}$.
\end{corollary}

\begin{proof}
Choose $t_i^A$ as the first time at which $|\sigma|=A_i$. Continuity gives $t_i^\varepsilon\le t_i^A$. Since $\alpha\ge\beta$, $\tau$ is nondecreasing. Since $\beta\in[0,1]$, so is $A/M(A)$, and therefore $\tau_a$. Hence $t_{i+1}-t_i^\varepsilon\ge t_{i+1}-t_i^A\ge\max\{\tau(A_i),\tau_a(A_i)\}\ge\overline\tau_\varepsilon$.
\end{proof}

Corollary~\ref{cor:threshold} gives an online lower bound once $|\sigma|$ first reaches the selected level $\varepsilon$ after a zero crossing. The quantity $\tau_\varepsilon$ is the component that remains uniform with respect to $a(t)$.

\begin{corollary}[Finite number of non-negligible excursions]
\label{cor:counting}
Let $[T_a,T_b]$ be a finite interval and let $N_\varepsilon$ denote the number of complete zero-crossing intervals contained in $[T_a,T_b]$ whose excursion amplitudes satisfy $A_i\ge\varepsilon>0$. Define $\Theta_\varepsilon
:=
\tau_r(\varepsilon)+\overline\tau_\varepsilon$. Then, $N_\varepsilon
\le
\left\lfloor
\frac{T_b-T_a}{\Theta_\varepsilon}
\right\rfloor
\le
\left\lfloor
\frac{T_b-T_a}{\tau_\varepsilon}
\right\rfloor$. Hence, only finitely many zero-crossing intervals whose excursion amplitude exceeds a prescribed nonzero threshold can occur over any finite time interval.
\end{corollary}

\begin{proof}
The functions $\tau$ and $\tau_a$ are nondecreasing. For $\tau_r$, use $\beta=2\alpha-1$ and write $\tau_r(A)=A^{1-\alpha}y(A)$, where the positive root $y(A)$ satisfies $\frac{1}{a_M}=k_1y+\frac{1}{2}(k_2+\delta A^{-\beta})y^2$.  The coefficient $k_2+\delta A^{-\beta}$ is nonincreasing, so $y(A)$ is nondecreasing. Since $1-\alpha\ge0$, $A^{1-\alpha}$ is also nondecreasing, and hence so is $\tau_r$. Thus every interval with $A_i\ge\varepsilon$ has length at least $\Theta_\varepsilon$.

For any finite subcollection of $m$ such intervals, disjointness gives $m\Theta_\varepsilon\le T_b-T_a$, hence $m\le\left\lfloor(T_b-T_a)/\Theta_\varepsilon\right\rfloor$. An infinite family would contain a finite subcollection larger than this bound, which is impossible. Therefore the same bound holds for $N_\varepsilon$. The second inequality follows from $\Theta_\varepsilon\ge\tau_\varepsilon$.
\end{proof}

Equation~\eqref{eq:integral-change} also gives a cumulative estimate.

\begin{corollary}[Cumulative excursion estimate]
\label{cor:cumulative}
Consider complete zero-crossing intervals $[t_i,t_{i+1}]$ contained in $[T_a,T_b]$. When $\beta=0$, the integrand $|\sigma|^\beta$ below is understood as one. Its value at isolated zeros does not affect the integral. Their excursion amplitudes satisfy
\begin{equation}
\label{eq:cumulative-general}
k_1\sum_i A_i^\alpha
\le
k_2\int_{T_a}^{T_b}|\sigma(t)|^\beta\,dt
+\delta(T_b-T_a).
\end{equation}
If $|\sigma(t)|\le\rho$ on $[T_a,T_b]$ for some $\rho>0$, then
\begin{equation}
\label{eq:cumulative-rho}
\sum_i A_i^\alpha
\le
\frac{k_2\rho^\beta+\delta}{k_1}(T_b-T_a).
\end{equation}
\end{corollary}

\begin{proof}
Sum \eqref{eq:integral-change} over any finite subcollection $\mathcal F$. The return segments $[t_i^A,t_{i+1}]$ are disjoint, so
$k_1\sum_{i\in\mathcal F}A_i^\alpha\le k_2\int_{T_a}^{T_b}|\sigma(t)|^\beta\,dt+\delta(T_b-T_a)$. Taking the supremum over all finite subcollections gives \eqref{eq:cumulative-general}. If $|\sigma|\le\rho$, then $\int_{T_a}^{T_b}|\sigma|^\beta dt\le\rho^\beta(T_b-T_a)$, which gives \eqref{eq:cumulative-rho}.
\end{proof}

\begin{corollary}[Behavior of an accumulating zero-crossing sequence]
\label{cor:accumulation}
Suppose that infinitely many complete zero-crossing intervals are contained in a finite interval. Then
\begin{equation}
\label{eq:alpha-summability}
\sum_i A_i^\alpha<\infty,
\end{equation}
and therefore $A_i\to0$. 
\end{corollary}

\begin{proof}
Equation~\eqref{eq:cumulative-general} gives \eqref{eq:alpha-summability} on every finite interval. Since the terms are nonnegative, $A_i^\alpha\to0$, and therefore $A_i\to0$.
\end{proof}

\begin{corollary}[Effective input along an accumulating zero-crossing sequence]
\label{cor:effective_input_accumulation}
Let $\{t_i\}_{i\in\mathbb{N}}$ be a sequence of consecutive sign-changing zero crossings of $\sigma$ such that $t_i\to t_\star<\infty$. Then $(\sigma(t_\star),z(t_\star))=(0,0)$, and $z(t_i)=u(t_i)+\phi(t_i)\to0$. Moreover, $u(t_i)\to-\phi(t_\star)$.
\end{corollary}

\begin{proof}
Continuity gives $\sigma(t_\star)=0$. Let $q_i$ be the sign of $\sigma$ just after $t_i$, and let $s_i:=\inf\{t\in(t_i,t_{i+1}]:\sigma(t)=0\}$. Then $q_i z(t_i)\ge0$ by Proposition~\ref{prop:crossing-sign}, while $q_i z(s_i)\le0$. Otherwise, $q_iw(s_i)>0$ and $q_i\sigma$ would increase just before its first return to zero.

Since $t_i\le s_i\le t_{i+1}$, both $t_i$ and $s_i$ converge to $t_\star$. Therefore $|z(t_i)|\le|z(t_i)-z(s_i)|\to0$, and continuity gives $z(t_\star)=0$. Finally, at a zero crossing $u(t_i)=\nu(t_i)$, so $z(t_i)=u(t_i)+\phi(t_i)\to0$. Since $\phi$ is continuous, $u(t_i)\to-\phi(t_\star)$.
\end{proof}

\begin{remark}[Finite-time accumulation and sliding-mode interpretation]
\label{rem:meaning_main_result}
The obtained bounds do not exclude finite-time accumulation of zero crossings. This can occur for the unperturbed classical STA as the trajectory reaches the origin \cite{b:SeeberHornFridman2018ReachingTime}. If such an accumulation occurs, Corollary~\ref{cor:accumulation} gives $A_i\to0$ and, more strongly, $\sum_i A_i^\alpha<\infty$ on every compact interval. Moreover, Corollary~\ref{cor:effective_input_accumulation} gives $z(t_i)\to0$ and $u(t_i)\to-\phi(t_\star)$ along the accumulating crossing sequence. These conclusions hold for the whole homogeneous family considered here.

The zero-crossing bounds do not impose a positive dwell time on the discontinuous term $\spow{\sigma}{0}$ driving the integral state of the classical STA. At $\sigma=0$, this term is set-valued in the Filippov sense, and its internal selection is not constrained by these bounds. What is excluded is the finite-time accumulation of zero crossings whose excursions remain bounded away from zero.

For the classical STA, $\ell=-1$ and $\beta=0$, the limiting behavior above is consistent with ideal second-order sliding. At $\sigma=0$, the Filippov selection can compensate $d(t)$ and maintain $(\sigma,z)=(0,0)$ with appropriate gain selection. For $\ell>-1$, however, $\beta>0$ and $\dot z=d(t)$ at $\sigma=0$. Therefore, $u(t_i)\to-\phi(t_\star)$ is only a limiting statement along the crossing sequence and does not, in general, imply that $(\sigma,z)=(0,0)$ remains invariant after the accumulation time.
\end{remark}

As an additional verification step, the algebraic theorem chain underlying the excursion-time, cumulative, and accumulation results was formalized in Lean~4.34.1 and machine-checked with Mathlib v4.34.1 (\url{https://github.com/ERC-INJECT/Zero-cross-super-twisting/tree/main/lean}). The trajectory-level Filippov arguments are treated analytically in the proofs above.

\section{Numerical illustration}
\label{sec:numerics}

The estimates are illustrated for the classical STA, corresponding to $\ell=-1$, $\alpha=1/2$, and $\beta=0$. The time-varying input coefficient and matched perturbation are selected as $a(t)=1+0.1\sin(0.7t)$ and $\phi(t)=0.2\sin(5t)$, so that $0.9\le a(t)\le a_M=1.1$ and $|\dot\phi(t)|\le\delta=1$. Following \cite{b:Mathey-Moreno-2024}, the gains are chosen as $k_1=1.5$ and $k_2=1.1$. The initial condition is $(\sigma(0),\nu(0))=(1,-1)$, and the simulation horizon is $t\in[0,7]$.

The physical coordinates $(\sigma,\nu)$ are integrated with an event-based BDF implementation of \texttt{solve\_ivp}, using relative tolerance $5\times10^{-7}$ and componentwise absolute tolerance $10^{-10}$. Zero crossings are detected through the non-terminal event $\sigma=0$, with \texttt{terminal=False} and both crossing directions enabled. Hence, the event logic records the crossings without interrupting or restarting the nominal integration. The solver reaches the prescribed final time after $3570$ accepted steps and $15063$ right-hand-side evaluations. All simulations were performed with Python 3.12.3, NumPy 1.26.4 and SciPy 1.14.1.

For each complete interval $[t_i,t_{i+1}]$, the sign $q_i$ is determined from the interior trajectory. The excursion maximum $t_i^A$ is reconstructed by locating a positive-to-negative zero of $q_iw(t)$ from the dense-output solution, rather than by directly maximizing the interpolated trajectory. Since $\dot\sigma=a(t)w(t)$ with $a(t)>0$, this procedure enforces the necessary condition $w(t_i^A)=0$. The amplitude resolution is set to $A_{\rm res}=10\,\mathrm{atol}_{\sigma}=10^{-9}$. This threshold is used only to identify excursions that are well separated from the selected amplitude resolution. No detected zero crossing is discarded.

Figure~\ref{fig:numerical-trajectories} shows the closed-loop trajectories. The simulation detects $273$ zero crossings, yielding $272$ complete inter-crossing intervals. A dynamically consistent maximum $t_i^A$ is reconstructed for $215$ intervals. For the remaining $57$, no positive-to-negative zero of $q_iw$ is resolved at the selected post-processing accuracy. After approximately $t=2$, the trajectory has reached the numerical sliding layer (for the exact Filippov dynamics, a trajectory on $(\sigma,z)=(0,0)$ remains on the sliding set), so subsequent detected sign changes are solver-level events due to finite tolerances rather than physical post-reaching oscillations. The variables $z$ and $w$ approach zero, while the continuous control approaches the matched-perturbation compensation value.
\begin{figure}[ht!]
  \centering
  \includegraphics[width=8.5cm,keepaspectratio]
  {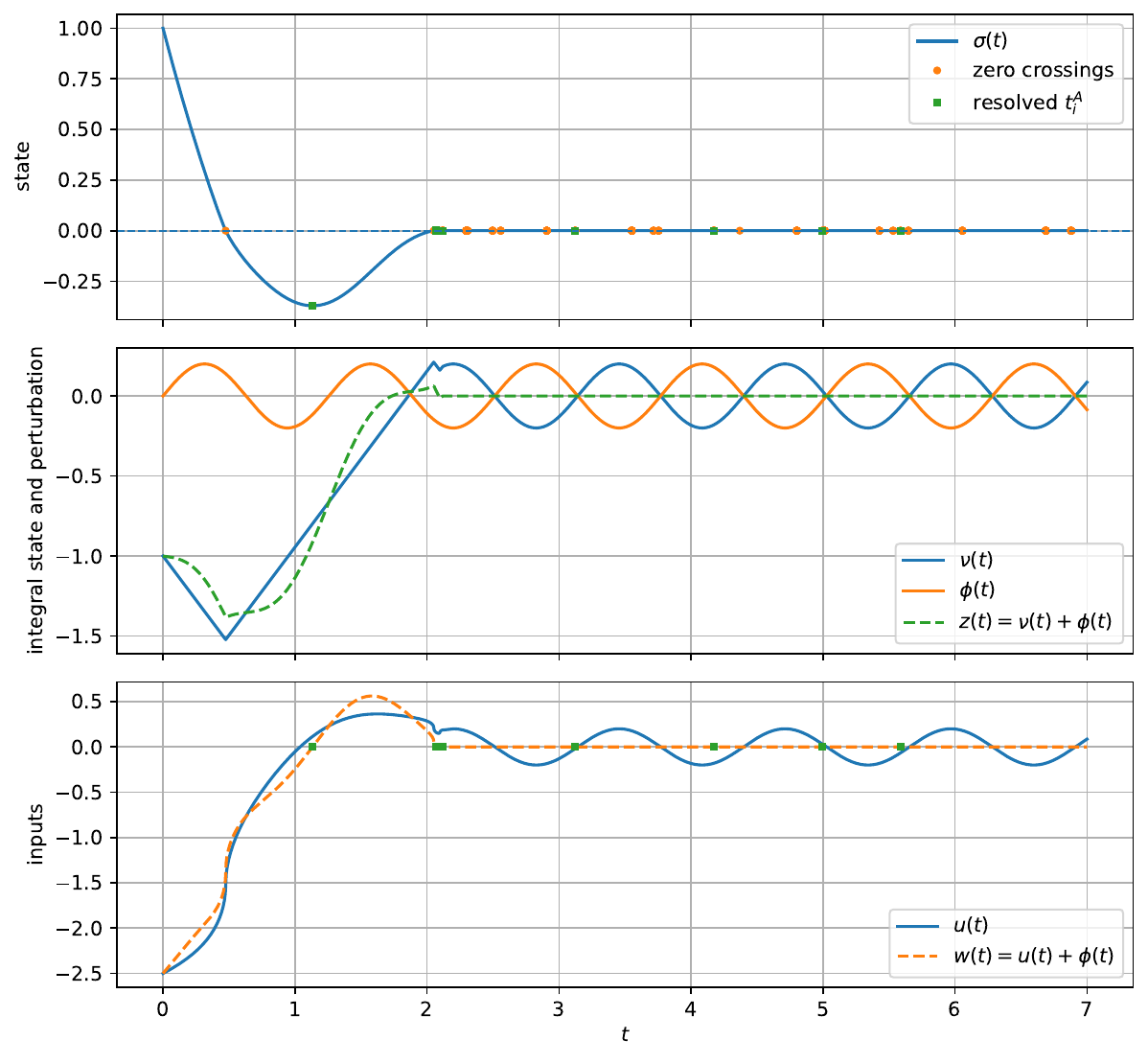}
  \caption{Closed-loop trajectories and detected zero crossings. Top:
  sliding variable $\sigma$, all passive zero-crossing events, and the
  reconstructed maxima $t_i^A$. Middle: controller integral state $\nu$,
  matched perturbation $\phi$, and disturbance-compensated coordinate
  $z=\nu+\phi$. Bottom: continuous control $u$ and effective input
  $w=u+\phi$. All detected crossings are retained.}
  \label{fig:numerical-trajectories}
\end{figure}

For each reconstructed interval, the return, rise, and complete-interval
ratios are defined as $\Gamma_i^{\rm ret}:=
\frac{t_{i+1}-t_i^A}
{\max\{\tau(A_i),\tau_a(A_i)\}}$, $\Gamma_i^{\rm rise}:=
\frac{t_i^A-t_i}{\tau_r(A_i)}$, and $\Gamma_i^{\rm int}:=
\frac{t_{i+1}-t_i}{\Theta(A_i)}$. The maximum normalized residual at the reconstructed turning points is
$1.052\times10^{-13}$, confirming $w(t_i^A)=0$ to numerical precision.
Nine reconstructed intervals satisfy $A_i\ge A_{\rm res}$. For this
resolved subset, the minimum ratios are
$\min_i\Gamma_i^{\rm ret}=1.3991$,
$\min_i\Gamma_i^{\rm rise}=1.9065$, and
$\min_i\Gamma_i^{\rm int}=1.5667$.

Figure~\ref{fig:numerical-excursion-bound} compares the measured return,
rise, and complete-interval times with the corresponding theoretical
lower bounds for these nine excursions. All three estimates are
satisfied in every case.
\begin{figure}[ht!]
  \centering
  \includegraphics[width=8.0cm,keepaspectratio]
  {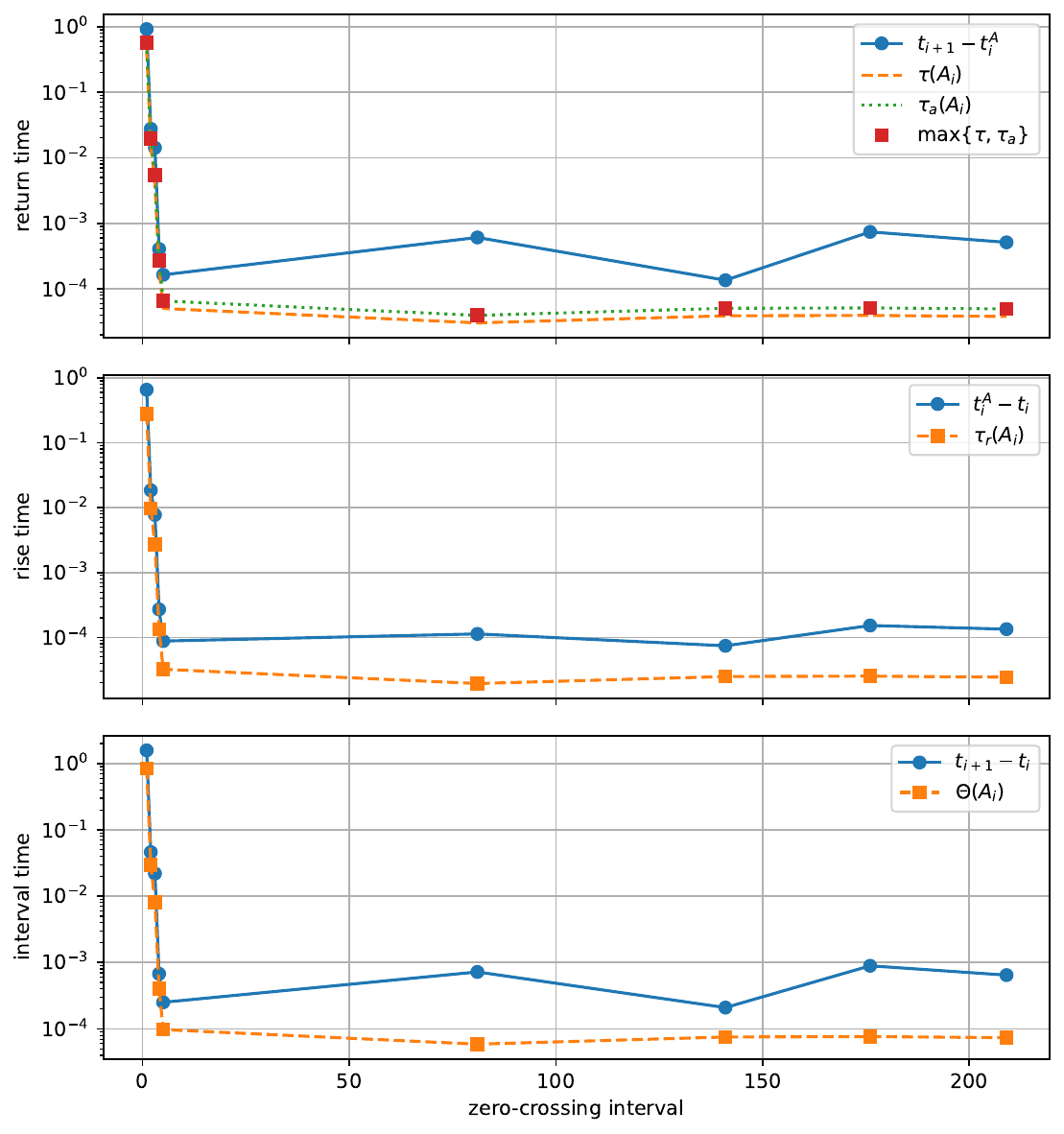}
  \caption{Excursion-dependent time estimates for the intervals satisfying
  $A_i\ge A_{\rm res}$. Top: measured return times and the bounds
  $\tau(A_i)$, $\tau_a(A_i)$, and $\max\{\tau(A_i),\tau_a(A_i)\}$.
  Middle: measured rise times and $\tau_r(A_i)$. Bottom: complete
  zero-crossing intervals and $\Theta(A_i)$.}
  \label{fig:numerical-excursion-bound}
\end{figure}

The online certificate of Corollary~\ref{cor:threshold} is evaluated for
$\varepsilon\in\{10^{-8},10^{-7},10^{-6},10^{-5},10^{-4}\}$. For each
excursion satisfying $A_i\ge\varepsilon$, the first outgoing time
$t_i^\varepsilon$ at which $|\sigma|=\varepsilon$ is located and the
ratio $\Gamma_i^\varepsilon:=
\frac{t_{i+1}-t_i^\varepsilon}{\overline\tau_\varepsilon}$ is computed. No violation is observed. Ordered by increasing
$\varepsilon$, the minimum ratios are $7.1305$, $74.5908$, $23.5176$,
$7.1550$, and $4.7715$. Thus, once the prescribed amplitude is reached,
the subsequent zero crossing occurs strictly after the certified
interval $\overline\tau_\varepsilon$.

The counting estimate of Corollary~\ref{cor:counting} is evaluated on the
complete crossing horizon $[T_a,T_b]=[0.47471755,6.88032770]$, whose
length is $6.40561015$. The observed excursion counts are
$N_\varepsilon=4,3,3,3,$ and $2$, whereas the corresponding upper bounds
based on $\Theta_\varepsilon$ are $46109$, $14580$, $4610$, $1458$, and
$461$. Figure~\ref{fig:numerical-online-certificate} confirms the
threshold-dependent time certificate and illustrates the expected
conservatism of the counting estimate.
\begin{figure}[ht!]
  \centering
  \includegraphics[width=8.5cm,keepaspectratio]
  {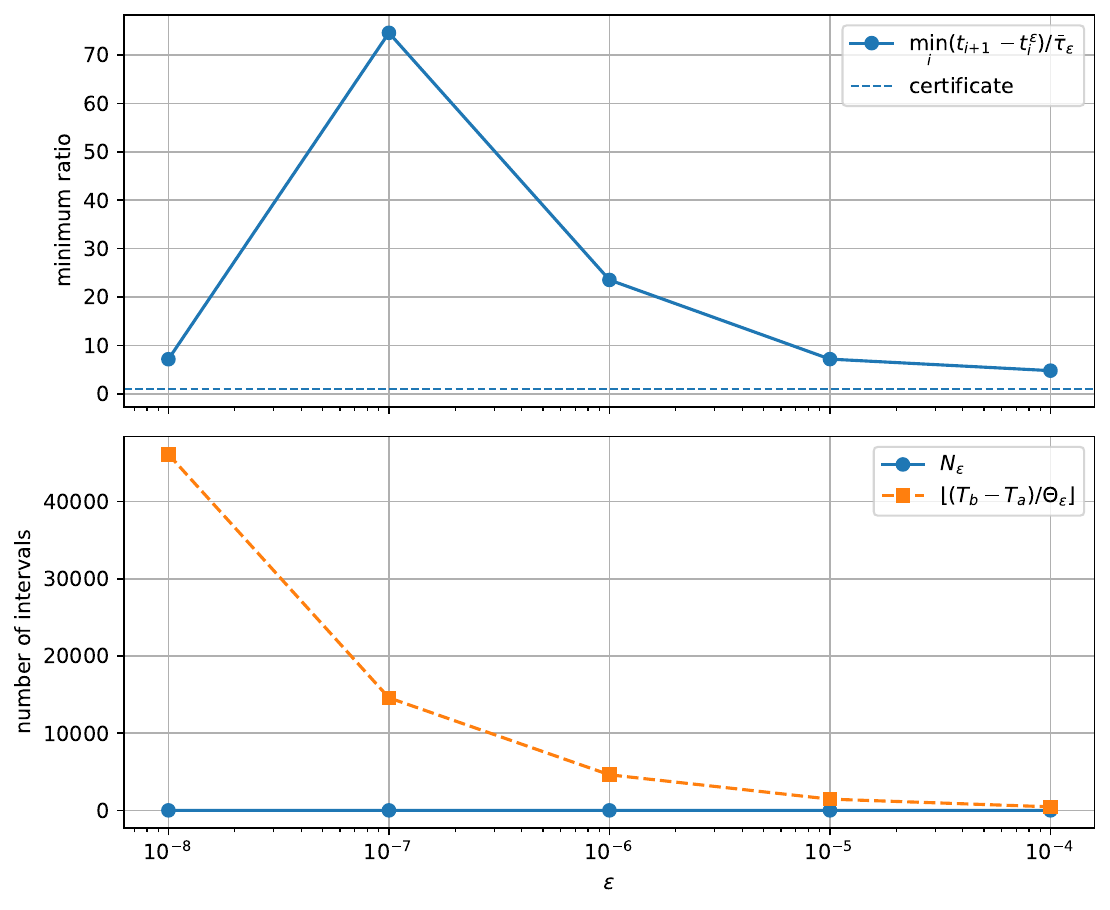}
  \caption{Prescribed-threshold certificate. Top: minimum ratio
  $\min_i\Gamma_i^\varepsilon$ for each threshold. The dashed line marks
  $\Gamma_i^\varepsilon=1$. Bottom: observed number $N_\varepsilon$ of
  complete excursions reaching the threshold and the upper bound based
  on $\Theta_\varepsilon$ from Corollary~\ref{cor:counting}.}
  \label{fig:numerical-online-certificate}
\end{figure}

For completeness, the continuous homogeneous case $\ell=-0.5$, for which
$\alpha=2/3$ and $\beta=1/3$, was also evaluated with
$k_1=5.6125$ and $k_2=4.1158$, while all remaining parameters and
numerical settings were kept unchanged. Again, the gain selection was made following \cite{b:Mathey-Moreno-2024}. The simulation produced $11$ zero
crossings and $10$ complete zero-crossing intervals, with a dynamically
consistent excursion maximum reconstructed in every interval. All three
excursion-dependent estimates were satisfied without violations, with
minimum return-, rise-, and complete-interval ratios equal to $2.1366$,
$3.4805$, and $2.3304$, respectively
(see Figure~\ref{fig:excursion_bound_ell05}). The prescribed-threshold
certificate was also verified for
$\varepsilon\in\{10^{-8},10^{-7},10^{-6},10^{-5},10^{-4}\}$.
All $10$ excursions reached each tested threshold, the corresponding
return-time ratios remained above one, and the counting bounds were
satisfied in every case (see Figure~\ref{fig:numerical-online-certificate_ell05}). These results
illustrate the dependence on $A_i^\beta$ and confirm the estimates beyond
the classical case $\beta=0$.
\begin{figure}[ht!]
  \centering
  \includegraphics[width=8.0cm,keepaspectratio]
  {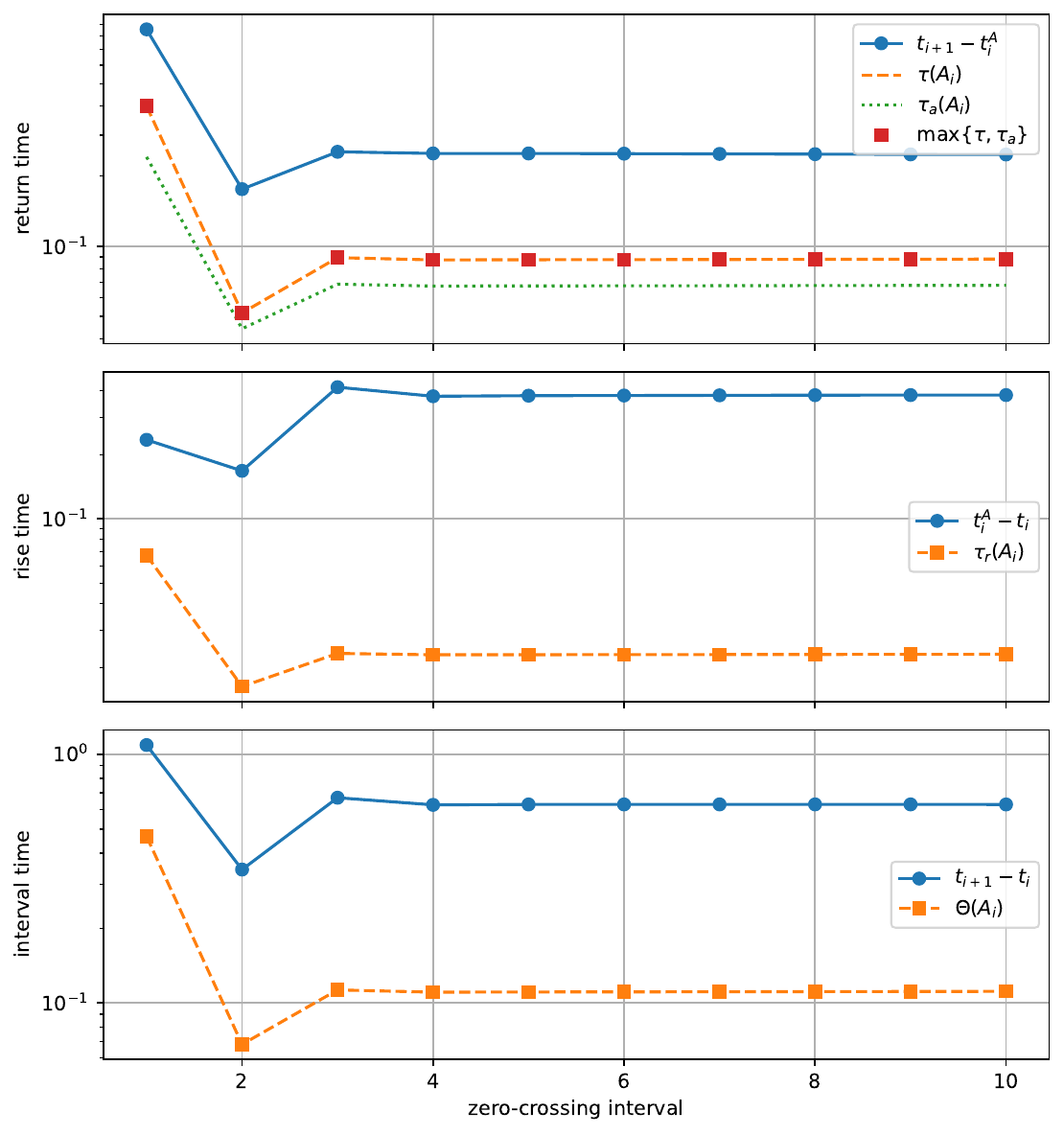}
  \caption{Excursion-dependent time estimates for $\ell=-0.5$.
  The return-, rise-, and complete-interval estimates are satisfied for
  all complete zero-crossing intervals, with minimum ratios $2.1366$,
  $3.4805$, and $2.3304$, respectively.}
  \label{fig:excursion_bound_ell05}
\end{figure}
\begin{figure}[ht!]
  \centering
  \includegraphics[width=8.5cm,keepaspectratio]
  {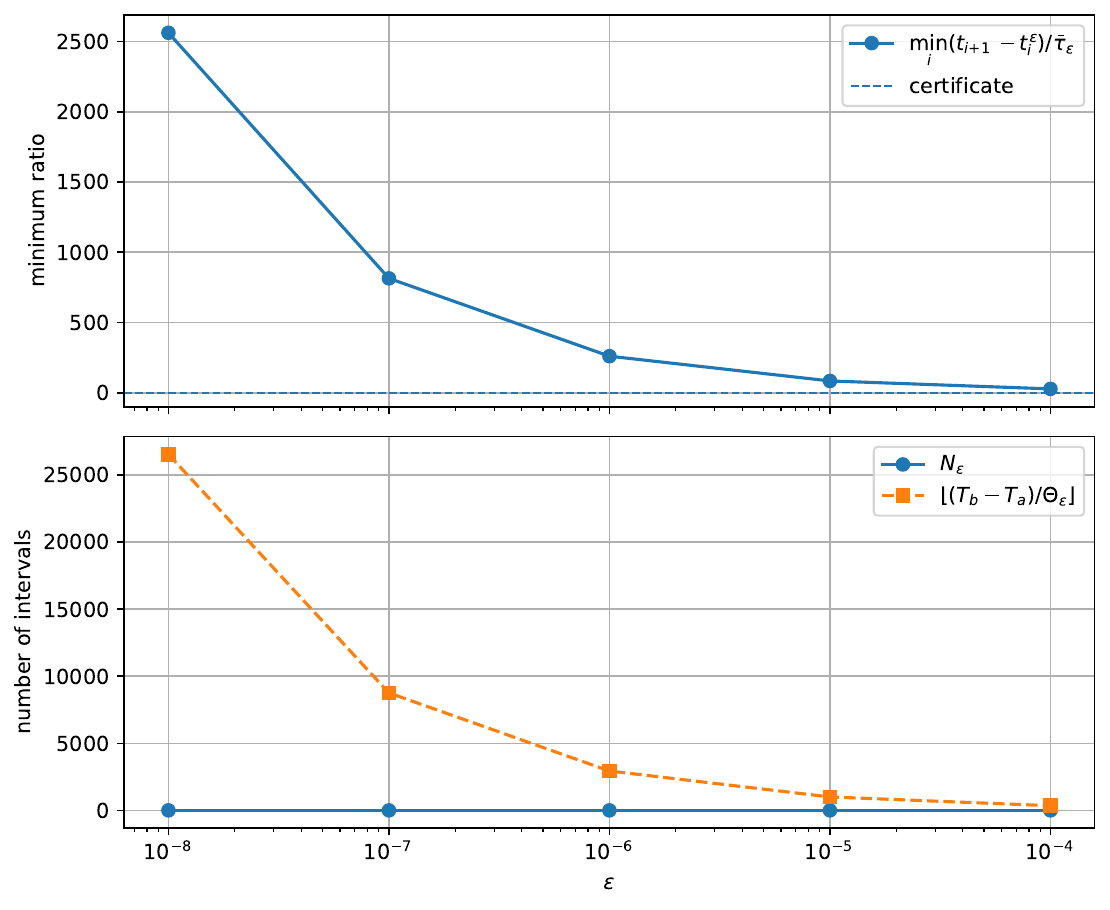}
  \caption{Prescribed-threshold certificate for $\ell=-0.5$. Top: minimum ratio
  $\min_i\Gamma_i^\varepsilon$ for each threshold. The dashed line marks
  $\Gamma_i^\varepsilon=1$. Bottom: observed number $N_\varepsilon$ of
  complete excursions reaching the threshold and the upper bound based
  on $\Theta_\varepsilon$ from Corollary~\ref{cor:counting}.}
  \label{fig:numerical-online-certificate_ell05}
\end{figure}

\vspace{-30pt}
\section{Discussion}
\label{sec:discussion}

Classical STA analyses primarily establish finite-time convergence, robustness margins, and reaching-time estimates \cite{b:levant_ijc1993,b:levant_automatica2005,b:LEVANT2007576,b:PolyakovPoznyak2009ReachingTime,b:6144710,b:SeeberHornFridman2018ReachingTime}. Lyapunov-based approaches provide upper bounds on the time required to reach the sliding set \cite{b:PolyakovPoznyak2009ReachingTime,b:6144710}, while the unperturbed analysis in \cite{b:SeeberHornFridman2018ReachingTime} resolves the zero structure as part of the reaching-time computation. The present result is complementary: it does not provide another convergence proof, but a trajectory-level certificate on how rapidly the sliding variable can complete a finite excursion and return through the origin. This temporal information is not contained in a convergence-time estimate, since two trajectories with the same reaching-time guarantee may exhibit different zero-crossing patterns before reaching the sliding motion.


The estimates presented in this paper also differ from chattering analyses based on parasitic dynamics. Describing-function and harmonic-balance methods characterize sustained oscillations created by actuator bandwidth, delays, sampling, or other implementation effects \cite{b:BoikoFridman2005,b:RosalesFreidovichCastillo2023,b:PEREZVENTURA2019175}. Here, no additional dynamics are introduced. The result therefore provides a continuous-time baseline: under the stated assumptions, zero crossings associated with a prescribed nonzero excursion cannot occur with arbitrarily small separation. Oscillations observed in an implementation must instead be analyzed in the corresponding actuator-, sampling-, or measurement-augmented model, since such effects are outside the scope of the present estimate.

The distinction from event-triggered control is equally important. Minimum inter-event times in event-triggered STA designs concern update instants generated by a selected triggering law and are used to guarantee implementability and exclude Zeno executions \cite{b:LiEtAl2020EventTriggeredSTA,b:KumarMohantyKumar2020EventTriggeredSTA,b:DouDingYu2023EventTriggeredSOSM}. In contrast, the present events are intrinsic sign changes of the sliding variable, and the controller is neither sampled nor modified. The threshold $\varepsilon$ is not a triggering parameter for the control law. It selects the smallest excursion regarded as relevant and converts it into a guaranteed separation between trajectory zero crossings.

A practical advantage is that the bounds depend only on the excursion amplitude of $\sigma$, the controller gains, and the known uncertainty bounds. The disturbance-compensated coordinate $z=\nu+\phi$ appears in the proof but not in the final estimate, and neither the perturbation amplitude nor the crossing velocity must be measured or reconstructed. Continuity alone implies that the amplitudes of any finite accumulating sequence must vanish, so excursions bounded away from zero cannot accumulate in finite time. The cumulative estimate provides the stronger quantitative statement $\sum_i A_i^\alpha<\infty$ on compact intervals. At a finite accumulation point, Corollary~\ref{cor:effective_input_accumulation} characterizes the limiting behavior of the effective input and the continuous control along the accumulating crossing sequence.

\section{Conclusion}
\label{sec:conclusion}

This work established lower bounds on the time between zero crossings of the scalar homogeneous STA, based on the amplitude of the corresponding excursion. The estimates provide minimum separation times and explicit upper bounds on the number of complete excursions above a prescribed amplitude. Although infinitely many excursions may still occur in finite time, their amplitudes must vanish and satisfy a summability condition. The paper also shows that, when zero crossings accumulate in finite time, the combined effect of the control and the matched perturbation approaches zero. Along these crossings, the control approaches the value needed to compensate the perturbation.

These results provide a quantitative description of fast zero-crossing activity in the STA, complementing existing stability, reaching-time, and chattering analyses. Extensions to multivariable, adaptive-gain, actuator, and sampled-data settings remain for future work.

\bmsubsection*{Author Contributions}

\textbf{Diego Guti\'errez-Oribio:} Conceptualization, Formal Analysis, Methodology, Validation, Visualization, Writing -- Original Draft Preparation. \textbf{Ioannis Stefanou:} Conceptualization, Formal Analysis, Funding Acquisition, Methodology, Project Administration, Writing -- Review \& Editing.

\bmsubsection*{Acknowledgments}
Funded by the European Union. Views and opinions expressed are, however, those of the author(s) only and do not necessarily reflect those of the European Union or the European Research Council Executive Agency. Neither the European Union nor the granting authority can be held responsible for them. This work is supported by the ERC grant INJECT, no. 101087771, doi: 10.3030/101087771.

\bmsubsection*{Financial Disclosure}

None reported.

\bmsubsection*{Conflicts of Interest}

The authors declare no conflicts of interest.

\bmsubsection*{Declaration of generative AI and AI-assisted technologies}

During the preparation of this work, the authors used ChatGPT for language editing and to assist in the construction of the Lean~4 formalization. The authors reviewed and edited the generated material before inclusion in the manuscript. 

\bibliography{Bibliografias}

\end{document}